\documentclass[conference, twocolumn]{ieeeconf}
\usepackage[T1]{fontenc}
\usepackage[latin9]{inputenc}
\usepackage{array}
\usepackage{verbatim}
\usepackage{booktabs}
\usepackage{units}
\usepackage{mathrsfs}
\usepackage{amsmath}
\usepackage{amssymb}
\usepackage{graphicx}
\usepackage{esint}
\usepackage{xargs}[2008/03/08]

\makeatletter

\providecommand{\tabularnewline}{\\}

\usepackage{amsthm}
  \newtheorem{defn}{\protect\definitionname}
\theoremstyle{definition}
\newtheorem{assumption}{Assumption}
  \newtheorem{thm}{\protect\theoremname}

\usepackage{cite} 
\IEEEoverridecommandlockouts
\usepackage{algpseudocode}

\newsavebox{\ieeealgbox}

\makeatother

\providecommand{\definitionname}{Definition}
\providecommand{\theoremname}{Theorem}

\begin{document}

\title{Simultaneous State and Parameter Estimation for Second-Order Nonlinear
Systems}

\author{Rushikesh Kamalapurkar\thanks{Rushikesh Kamalapurkar is with the School of Mechanical and Aerospace
Engineering, Oklahoma State University, Stillwater, OK, USA. {\tt\small rushikesh.kamalapurkar@okstate.edu}.}}
\maketitle
\begin{abstract}
In this paper, a concurrent learning based adaptive observer is developed
for a class of second-order nonlinear time-invariant systems with
uncertain dynamics. The developed technique results in uniformly ultimately
bounded state and parameter estimation errors. As opposed to \textit{persistent}
excitation which is required for parameter convergence in traditional
adaptive control methods, the developed technique only requires excitation
over a finite time interval to achieve parameter convergence. Simulation
results in both noise-free and noisy environments are presented to
validate the design.
\end{abstract}

\section{Introduction}

\global\long\def\d{\textnormal{d}}

\global\long\def\i{\textnormal{i}}

\global\long\def\sgn{\operatorname{sgn}}

\global\long\def\proj{\operatornamewithlimits{proj}}

\global\long\def\argmin{\operatornamewithlimits{arg\:min}}

\global\long\def\r{\mathbb{R}}

\global\long\def\tr{\operatorname{tr}}

\global\long\def\Sgn{\operatorname{SGN}}

\global\long\def\k{\operatorname{K}}

\global\long\def\co{\operatorname{co}}

\global\long\def\ae{\operatorname{a.e.}}

\global\long\def\b{\operatorname{B}}

\global\long\def\n{\mathbb{N}}

\global\long\def\diag{\operatorname{diag}}

\global\long\def\rank{\operatorname{rank}}

\global\long\def\ind{\operatorname{\mathbf{1}}}

\newcommandx\fid[5][usedefault, addprefix=\global, 1=, 2=, 3=, 4=, 5=]{\tensor*[_{#3}^{#2}]{#1}{_{#5}^{#4}}}

\global\long\def\g{\operatorname{g}}

\global\long\def\eval#1{\left.#1\right|}

\global\long\def\lab#1#2{\underset{#2}{\underbrace{#1}}}

\global\long\def\zero{\operatorname{0}}

\global\long\def\e{\textnormal{e}}

\global\long\def\vecop{\operatorname{vec}}

\global\long\def\id{\operatorname{I}}

\global\long\def\re{\operatorname{re}}

\global\long\def\im{\operatorname{im}}

\global\long\def\lap#1{\mathscr{L}\left\{  #1\right\}  }

\global\long\def\dist{\operatorname{dist}}

\global\long\def\unit#1#2{\SI[per-mode=symbol]{#1}{#2}}

Owing to Increasing reliance on
automation and increasing complexity of autonomous systems, the ability
to adapt has become an indispensable feature of modern control systems.
Traditional adaptive control methods (see, e.g., \cite{SCC.Ioannou.Sun1996,SCC.Sastry.Bodson1989,SCC.Krstic.Kanellakopoulos.ea1995})
attempt to improve the tracking performance, and in general, do not
focus on parameter estimation. While accurate parameter estimation
can improve robustness and transient performance of adaptive controllers,
(see, e.g., \cite{SCC.Duarte.Narendra1989,SCC.Krstic.Kokotovic.ea1993,SCC.Chowdhary.Johnson2011a}),
parameter convergence typically requires restrictive assumptions such
as persistence of excitation. An excitation signal is often added
to the controller to ensure persistence of excitation; however, the
added signal can cause mechanical fatigue and compromise the tracking
performance.

Parameter convergence can be achieved under a finite excitation condition
using data-driven methods such as concurrent learning (see, e.g.,
\cite{SCC.Chowdhary.Johnson2011a,SCC.Chowdhary.Yucelen.ea2013,SCC.Kersting.Buss2014}),
where the parameters are estimated by storing data during time-intervals
when the system is excited, and then utilizing the stored data to
drive adaptation when excitation is unavailable. Concurrent learning
has been shown to be an effective tool for adaptive control (see,
e.g., \cite{SCC.Chowdhary.Johnson2011a,SCC.Chowdhary.Muehlegg.ea2013,SCC.Chowdhary.Yucelen.ea2013,SCC.Kersting.Buss2014})
and adaptive estimation (see, e.g., \cite{SCC.Modares.Lewis.ea2014,SCC.Kamalapurkar.Klotz.ea2014a,SCC.Luo.Wu.ea2014,SCC.Kamalapurkar.Walters.ea2016,SCC.Bian.Jiang2016,SCC.Kamalapurkar.Rosenfeld.ea2016}),
however, concurrent learning typically requires full state feedback
along with accurate numerical estimates of the state-derivative.

Novel concurrent learning techniques that can be implemented using
full state measurements but without numerical estimates of the state-derivative
are developed in \cite{SCC.Kamalapurkar.Reish.ea2017} and \cite{SCC.Parikh.Kamalapurkar.easubmitteda};
however, since full state feedback is typically not available, the
development of an output-feedback concurrent learning framework is
well-motivated. An output feedback concurrent learning technique is
developed for second-order linear systems in \cite{SCC.Kamalapurkar2017};
however, the implementation critically depends on the certainty equivalence
principle, and hence, is not directly transferable to nonlinear systems.

In this paper, an output feedback concurrent learning method is developed
for simultaneous state and parameter estimation in second-order uncertain
nonlinear systems. An adaptive state-observer is utilized to generate
estimates of the state from input-output data. The estimated state
trajectories along with the known inputs are then utilized in a novel
data-driven parameter estimation scheme to achieve simultaneous state
and parameter estimation. Convergence of the state estimates and the
parameter estimates to a small neighborhood of the origin is established
under a finite (as opposed to \textit{persistent}) excitation condition.

The paper is organized as follows. An integral error system that facilitates
parameter estimation is developed in Section \ref{sec:Error-System-for}.
Section \ref{sec:Velocity-estimator-design} is dedicated to the design
of a robust state observer. Section \ref{sec:Parameter-Estimator-Design}
details the developed parameter estimator. Section \ref{sec:Purging}
details the algorithm for selection and storage of the data that is
used to implement concurrent learning. Section \ref{sec:Stability-Analysis}
is dedicated to a Lyapunov-based analysis of the developed technique.
Section \ref{sec:Simulation} demonstrates the efficacy of the developed
method via a numerical simulation and Section \ref{sec:Conclusion}
concludes the paper.

\section{\label{sec:Error-System-for}Error System for Estimation}

Consider a second order nonlinear system of the form\footnote{For $a\in\r,$ the notation $\r_{\geq a}$ denotes the interval $\left[a,\infty\right)$
and the notation $\r_{>a}$ denotes the interval $\left(a,\infty\right)$.}
\begin{align}
\dot{p}\left(t\right) & =q\left(t\right),\nonumber \\
\dot{q}\left(t\right) & =f\left(x\left(t\right),u\left(t\right)\right),\nonumber \\
y\left(t\right) & =p\left(t\right),\label{eq:Linear System}
\end{align}
where $p:\mathbb{R}_{\geq T_{0}}\to\r^{n}$ and $q:\r_{\geq T_{0}}\to\r^{n}$
denote the generalized position states and the generalized velocity
states, respectively, $x\triangleq\begin{bmatrix}p^{T} & q^{T}\end{bmatrix}^{T}$
is the system state, $f:\r^{n}\times\r^{m}\to\r^{n}$ is locally Lipschitz
continuous, and $y:\r_{\geq T_{0}}\to\r^{n}$ denotes the output.
The model $f$ is comprised of a known nominal part and an unknown
part, i.e., $f=f^{o}+g$, where $f^{o}:\r^{n}\times\r^{m}\to\r^{n}$
is known and locally Lipschitz and $g:\r^{n}\times\r^{m}\to\r^{n}$
is unknown and locally Lipschitz. The objective is to design an adaptive
estimator to identify the unknown function $g$, online, using input-output
measurements. It is assumed that the system is controlled using a
stabilizing input, i.e., $x,\:u\in\mathcal{L}_{\infty}$. It is further
assumed that the signal $p$, and $u$ are available for feedback.
Systems of the form (\ref{eq:Linear System}) encompass second-order
linear systems and Euler-Lagrange models, and hence, represent a wide
class of physical plants, including but not limited to robotic manipulators
and autonomous ground, aerial, and underwater vehicles. 

Given a compact set $\chi\subset\r^{n}\times\r^{m}$, and a constant
$\overline{\epsilon}$, the unknown function $g$ can be approximated
using basis functions as $g\left(x,u\right)=\theta^{T}\sigma\left(x,u\right)+\epsilon\left(x,u\right)$,
where $\sigma:\r^{n}\times\r^{m}\to\r^{p}$ and $\epsilon:\r^{n}\times\r^{m}\to\r^{n}$
denote the basis vector and the approximation error, respectively,
$\theta\in\r^{p\times n}$ is a constant matrix of unknown parameters,
and there exist $\overline{\sigma},\overline{\theta}>0$ such that
$\sup_{\left(x,u\right)\in\chi}\sigma\left(x,u\right)<\overline{\sigma}$,
$\sup_{\left(x,u\right)\in\chi}\nabla\sigma\left(x,u\right)<\overline{\sigma}$,
$\sup_{\left(x,u\right)\in\chi}\epsilon\left(x,u\right)<\overline{\epsilon}$,
$\sup_{\left(x,u\right)\in\chi}\nabla\epsilon\left(x,u\right)<\overline{\epsilon}$,
and $\left\Vert \theta\right\Vert <\overline{\theta}$. To obtain
an error signal for parameter identification, the system in (\ref{eq:Linear System})
is expressed in the form
\begin{equation}
\ddot{q}\left(t\right)=f^{o}\left(x\left(t\right),u\left(t\right)\right)+\theta^{T}\sigma\left(x\left(t\right),u\left(t\right)\right)+\epsilon\left(x\left(t\right),u\left(t\right)\right).\label{eq:Pre Integral Form}
\end{equation}
Integrating (\ref{eq:Pre Integral Form}) over the interval $\left[t-\tau_{1},t\right]$
for some constant $\tau_{1}\in\r_{>0}$ and then over the interval
$\left[t-\tau_{2},t\right]$ for some constant $\tau_{2}\in\r_{>0}$,
\begin{equation}
\intop_{t-\tau_{2}}^{t}\left(q\left(\lambda\right)-q\left(\lambda-\tau_{1}\right)\right)\d\lambda=\mathcal{I}f^{o}\left(t\right)+\theta^{T}\mathcal{I}\sigma\left(t\right)+\mathcal{I}\epsilon\left(t\right),\label{eq:Double Integral Form}
\end{equation}
where $\mathcal{I}$ denotes the integral operator $f\mapsto\intop_{t-\tau_{2}}^{t}\intop_{\lambda-\tau_{1}}^{\lambda}f\left(x\left(\tau\right),u\left(\tau\right)\right)\d\tau\d\lambda$.
Using the Fundamental Theorem of Calculus and the fact that $q\left(t\right)=\dot{p}\left(t\right)$,
the expression in (\ref{eq:Derivative Free Form}) can be rearranged
to form the affine system
\begin{multline}
P\left(t\right)=F\left(t\right)+\text{\ensuremath{\theta}}^{T}G\left(t\right)+E\left(t\right),\:\forall t\in\r_{\geq T_{0}}\label{eq:Derivative Free Form}
\end{multline}
where 
\begin{equation}
P\left(t\right)\triangleq\begin{cases}
\begin{gathered}p\left(t\!-\!\tau_{2}\!-\!\tau_{1}\right)\!-\!p\left(t\!-\!\tau_{1}\right)\\
\!+p\left(t\right)\!-\!p\left(t\!-\!\tau_{2}\right),
\end{gathered}
 & t\!\in\!\left[T_{0}\!+\!\tau_{1}\!+\!\tau_{2},\infty\right),\\
0 & t<T_{0}+\tau_{1}+\tau_{2}.
\end{cases}\label{eq:P}
\end{equation}
\begin{equation}
F\left(t\right)\triangleq\begin{cases}
\mathcal{I}f^{o}\left(t\right), & t\in\left[T_{0}+\tau_{1}+\tau_{2},\infty\right),\\
0, & t<T_{0}+\tau_{1}+\tau_{2},
\end{cases}\label{eq:F}
\end{equation}
 
\begin{equation}
G\left(t\right)\!\triangleq\!\begin{cases}
\!\mathcal{I}\sigma\left(t\right), & t\!\in\!\left[T_{0}\!+\!\tau_{1}\!+\!\tau_{2},\infty\!\right),\\
0 & t<T_{0}+\tau_{1}+\tau_{2},
\end{cases}\label{eq:G}
\end{equation}
and 
\begin{equation}
E\left(t\right)\triangleq\begin{cases}
\mathcal{I}\epsilon\left(t\right), & t\in\left[T_{0}+\tau_{1}+\tau_{2},\infty\right),\\
0 & t<T_{0}+\tau_{1}+\tau_{2}.
\end{cases}\label{eq:U}
\end{equation}
The affine relationship in (\ref{eq:Derivative Free Form}) is valid
for all $t\in\r_{\geq T_{0}}$; however, it provides useful information
about the vector $\theta$ only after $t\geq T_{0}+\tau_{1}+\tau_{2}$. 

The knowledge of the generalized velocity, $q$, is required to compute
the matrices $F$ and $G$. In the following, a robust adaptive velocity
estimator is developed to generate estimates of the generalized velocity.

\section{\label{sec:Velocity-estimator-design}Velocity estimator design}

To generate estimates of the generalized velocity, a velocity estimator
inspired by \cite{SCC.Dinh.Kamalapurkar.ea2014} is developed. The
estimator is given by
\begin{align}
\dot{\hat{p}} & =\hat{q}\nonumber \\
\dot{\hat{q}} & =f^{o}\left(\hat{x},u\right)+\hat{\theta}^{T}\sigma\left(\hat{x},u\right)+\nu,\label{eq:Estimator}
\end{align}
where $\hat{x}$, $\hat{p}$, $\hat{q},$ and $\hat{\theta}$ are
estimates of $x,$ $p,$ $q,$ and $\theta$, respectively, and $\nu$
is a feedback term designed in the following.

To facilitate the design of $\nu$, let $\tilde{p}=p-\hat{p}$, $\tilde{q}=q-\hat{q}$,
$\tilde{\theta}=\theta-\hat{\theta}$, and let
\begin{equation}
r\left(t\right)=\dot{\tilde{p}}\left(t\right)+\alpha\tilde{p}\left(t\right)+\eta\left(t\right),\label{eq:r}
\end{equation}
where the signal $\eta$ is added to compensate for the fact that
the generalized velocity state, $q$, is not measurable. Based on
the subsequent stability analysis, the signal $\eta$ is designed
as the output of the dynamic filter 
\begin{align}
\dot{\eta}\left(t\right) & =-\beta\eta\left(t\right)-kr\left(t\right)-\alpha\tilde{q}\left(t\right),\quad\eta\left(T_{0}\right)=0,\label{eq:eta Update}
\end{align}
where $\alpha,$ $k,$ and $\beta$ are positive constants and the
feedback component $\nu$ is designed as 
\begin{equation}
\nu\left(t\right)=\alpha^{2}\tilde{p}\left(t\right)-\left(k+\alpha+\beta\right)\eta\left(t\right).\label{eq:Feedback}
\end{equation}
The design of the signals $\eta$ and $\nu$ to estimate the state
from output measurements is inspired by the $p-$filter \cite{SCC.Xian.Queiroz.ea2004}.
Using the fact that $\tilde{p}\left(T_{0}\right)=0$, the signal $\eta$
can be implemented via the integral form
\begin{equation}
\eta\left(t\right)=-\intop_{T_{0}}^{t}\left(\beta+k\right)\eta\left(\tau\right)\d\tau-\intop_{T_{0}}^{t}k\alpha\tilde{p}\left(\tau\right)\d\tau-\left(k+\alpha\right)\tilde{p}\left(t\right).\label{eq:IntegralUpdateEta}
\end{equation}

The affine error system in (\ref{eq:Derivative Free Form}) motivates
the adaptive estimation scheme that follows. The design is inspired
by the \textit{concurrent learning} technique \cite{SCC.Chowdhary2010}.
Concurrent learning enables parameter convergence in adaptive control
by using stored data to update the parameter estimates. Traditionally,
adaptive control methods guarantee parameter convergence only if the
appropriate PE conditions are met \cite[Chapter 4]{SCC.Ioannou.Sun1996}.
Concurrent learning uses stored data to soften the PE condition to
an excitation condition over a finite time-interval. Concurrent learning
methods such as \cite{SCC.Chowdhary.Johnson2011a} and \cite{SCC.Kersting.Buss2014}
require numerical differentiation of the system state, and concurrent
learning techniques such as \cite{SCC.Parikh.Kamalapurkar.easubmitteda}
and \cite{SCC.Kamalapurkar.Reish.ea2017} require full state measurements.
In the following, a concurrent learning method that utilizes only
the output measurements is developed. 

\section{\label{sec:Parameter-Estimator-Design}Parameter Estimator Design}

To obtain output-feedback concurrent learning update law for the parameter
estimates, a history stack, denoted by $\mathcal{H}$, is utilized.
The history stack is a set of ordered pairs $\left\{ \left(P_{i},\hat{F}_{i},\hat{G}_{i}\right)\right\} _{i=1}^{M}$
such that 
\begin{equation}
P_{i}=\hat{F}_{i}+\text{\ensuremath{\theta}}^{T}\hat{G}_{i}+\mathcal{E}_{i},\:\forall i\in\left\{ 1,\cdots,M\right\} ,\label{eq:History Stack Compatibility}
\end{equation}
where $\mathcal{E}_{i}$ is a constant matrix. If a history stack
that satisfies (\ref{eq:History Stack Compatibility}) is not available
a priori, it is recorded online, based on the relationship in (\ref{eq:Derivative Free Form}),
by selecting an increasing set of time-instances $\left\{ t_{i}\right\} _{i=1}^{M}$
and letting 
\begin{gather}
P_{i}=P\left(t_{i}\right),\quad\hat{F}_{i}=\hat{F}\left(t_{i}\right),\quad\hat{G}_{i}=\hat{G}\left(t_{i}\right),\label{eq:Online Recording}
\end{gather}
where 
\begin{equation}
\hat{F}\left(t\right)\triangleq\begin{cases}
\hat{\mathcal{I}}f^{o}\left(t\right), & t\in\left[T_{0}+\tau_{1}+\tau_{2},\infty\right),\\
0, & t<T_{0}+\tau_{1}+\tau_{2},
\end{cases}\label{eq:F-1}
\end{equation}
 
\begin{equation}
\hat{G}\left(t\right)\!\triangleq\!\begin{cases}
\hat{\mathcal{I}}\sigma\left(t\right) & t\!\in\!\left[T_{0}\!+\!\tau_{1}\!+\!\tau_{2},\infty\!\right),\\
0 & t<T_{0}+\tau_{1}+\tau_{2},
\end{cases}\label{eq:G-1}
\end{equation}
where $\mathcal{\hat{I}}$ denote the operator $f\mapsto\intop_{t-\tau_{2}}^{t}\intop_{\lambda-\tau_{1}}^{\lambda}f\left(\hat{x}\left(\tau\right),u\left(\tau\right)\right)\d\tau\d\lambda$.
In this case, the error term $\mathcal{E}_{i}$ is given by $\mathcal{E}_{i}=E\left(t_{i}\right)+F\left(t_{i}\right)-\hat{F}\left(t_{i}\right)+\theta^{T}\left(G\left(t_{i}\right)-\hat{G}\left(t_{i}\right)\right).$
Let $\left[t_{1},t_{2}\right)$ be an interval over which the history
stack was recorded. Provided the states and the state estimates remain
within a compact set $\chi$ over $I\triangleq\left[t_{1}-\tau_{1}-\tau_{2},t_{2}\right)$,
the error terms can be bounded as
\begin{equation}
\left\Vert \mathcal{E}_{i}\right\Vert \leq L_{1}\overline{\epsilon}+L_{2}\overline{\tilde{x}}_{I},\forall i\in\left\{ 1,\cdots,M\right\} ,\label{eq:ErrorBound}
\end{equation}
where $\overline{\tilde{x}}_{I}\triangleq\max_{i\in\left\{ 1,\cdots,M\right\} }\sup_{t\in I}\left\Vert \tilde{x}\left(t\right)\right\Vert $
and $L_{1},L_{2}>0$ are constants.

The concurrent learning update law to estimate the unknown parameters
is designed as
\begin{equation}
\dot{\hat{\theta}}\left(t\right)=k_{\theta}\Gamma\left(t\right)\sum_{i=1}^{M}\hat{G}_{i}\left(P_{i}-\hat{F}_{i}-\hat{\theta}^{T}\left(t\right)\hat{G}_{i}\right)^{T},\label{eq:Theta Dynamics}
\end{equation}
where $k_{\theta}\in\r_{>0}$ is a constant adaptation gain and $\Gamma:\r_{\geq0}\to\r^{\left(2n^{2}+mn\right)\times\left(2n^{2}+mn\right)}$
is the least-squares gain updated using the update law
\begin{equation}
\dot{\Gamma}\left(t\right)=\beta_{1}\Gamma\left(t\right)-k_{\theta}\Gamma\left(t\right)\mathscr{G}\Gamma\left(t\right).\label{eq:Gamma Dynamics}
\end{equation}
where the matrix $\mathscr{G}\in\r^{p\times p}$ is defined as $\mathscr{G}\triangleq\sum_{i=1}^{M}\hat{G}_{i}\hat{G}_{i}^{T}$.
Using arguments similar to \cite[Corollary 4.3.2]{SCC.Ioannou.Sun1996},
it can be shown that provided $\lambda_{\min}\left\{ \Gamma^{-1}\left(T_{0}\right)\right\} >0$,
the least squares gain matrix satisfies 
\begin{equation}
\underline{\Gamma}\id_{p}\leq\Gamma\left(t\right)\leq\overline{\Gamma}\id_{p},\label{eq:StaFGammaBound}
\end{equation}
where $\underline{\Gamma}$ and $\overline{\Gamma}$ are positive
constants, and $\id_{n}$ denotes an $n\times n$ identity matrix.

\section{\label{sec:Purging}Purging}

The update law in (\ref{eq:Theta Dynamics}) is motivated by the fact
that if the full state were available for feedback and if the approximation
error, $\epsilon$, were zero, then using $\left[\begin{array}{ccc}
P_{1} & \cdots & P_{n}\end{array}\right]^{T}=\left[\begin{array}{ccc}
F_{1} & \cdots & F_{n}\end{array}\right]^{T}+\left[\begin{array}{ccc}
G_{1} & \cdots & G_{n}\end{array}\right]^{T}\text{\ensuremath{\theta}},$ the parameters could be estimated via the least squares estimate
$\hat{\theta}_{\textnormal{LS}}=\mathscr{G}^{-1}\left[\begin{array}{ccc}
G_{1} & \cdots & G_{n}\end{array}\right]\left[\begin{array}{ccc}
P_{1} & \cdots & P_{n}\end{array}\right]^{T}-\mathscr{G}^{-1}\left[\begin{array}{ccc}
G_{1} & \cdots & G_{n}\end{array}\right]\left[\begin{array}{ccc}
F_{1} & \cdots & F_{n}\end{array}\right]^{T}$. However, since the history stack contains the estimated terms $\hat{F}$
and $\hat{G}$, during the transient period where the state estimation
error is large, the history stack does not accurately (within the
error bound introduced by $\epsilon$) represent the system dynamics.
Hence, the history stack needs to be purged whenever better estimates
of the state are available.

Since the state estimator exponentially drives the estimation error
to a small neighborhood of the origin, a newer estimate of the state
can be assumed to be at least as good as an older estimate. A dwell
time based greedy purging algorithm is developed in this paper to
utilize newer data for estimation while preserving stability of the
estimator.

The algorithm maintains two history stacks, a main history stack and
a transient history stack, labeled $\mathcal{H}$ and $\mathcal{G}$,
respectively. As soon as the transient history stack is full and sufficient
dwell time has passed, the main history stack is emptied and the transient
history stack is copied into the main history stack. The sufficient
dwell time, denoted by $\mathcal{T}$, is determined using a Lyapunov-based
stability analysis. 

\begin{comment}
\[
\eta\left(t\right)=-\intop_{T_{0}}^{t}\left(\left(\beta+k\right)\eta\left(\tau\right)+k\alpha\tilde{p}\left(\tau\right)\right)\d\tau-\left(k+\alpha\right)\tilde{p}\left(t\right).
\]
Let 
\[
\eta_{1}\left(t\right)\triangleq\intop_{T_{0}}^{t}\left(\left(\beta+k\right)\eta\left(\tau\right)+k\alpha\tilde{p}\left(\tau\right)\right)\d\tau.
\]
Then, 
\[
\dot{\eta}_{1}=\left(\beta+k\right)\eta\left(\tau\right)+k\alpha\tilde{p}\left(\tau\right),\quad\eta_{1}\left(T_{0}\right)=0
\]

Also,

\begin{multline*}
\hat{q}\left(t\right)=\intop_{T_{0}}^{t}\left(\left(u\left(\tau\right)\varotimes\id_{n}\right)^{T}\vecop\left(\hat{B}\left(\tau\right)\right)+\nu\left(\tau\right)+\left(p\left(\tau\right)\varotimes\id_{n}\right)^{T}\left(\vecop\left(\hat{A}_{1}\left(\tau\right)\right)-\vecop\left(\dot{\hat{A}}_{2}\left(\tau\right)\right)\right)\right)\d\tau\\
+\hat{q}\left(T_{0}\right)+\left(p\left(t\right)\varotimes\id_{n}\right)^{T}\vecop\left(\hat{A}_{2}\left(t\right)\right)-\left(p\left(T_{0}\right)\varotimes\id_{n}\right)^{T}\vecop\left(\hat{A}_{2}\left(T_{0}\right)\right)
\end{multline*}
Let $\hat{q}_{1}\left(t\right)=\left(u\left(\tau\right)\varotimes\id_{n}\right)^{T}\vecop\left(\hat{B}\left(\tau\right)\right)+\nu\left(\tau\right)+\left(p\left(\tau\right)\varotimes\id_{n}\right)^{T}\left(\vecop\left(\hat{A}_{1}\left(\tau\right)\right)-\vecop\left(\dot{\hat{A}}_{2}\left(\tau\right)\right)\right)$
\end{comment}
Parameter identification in the developed framework imposes the following
requirement on the history stack $\mathcal{H}$.
\begin{defn}
A history stack $\left\{ \left(P_{i},\hat{F}_{i},\hat{G}_{i}\right)\right\} _{i=1}^{M}$
is called \textit{full rank }if there exists a constant $\underline{c}\in\r$
such that 
\begin{equation}
0<\underline{c}<\lambda_{\min}\left\{ \mathscr{G}\right\} ,\label{eq:Rank Condition}
\end{equation}
where $\lambda_{\min}\left(\cdot\right)$ denotes the minimum singular
value of a matrix.
\end{defn}
\begin{assumption}
For a given $M\in\mathbb{N}$ and $\underline{c}\in\r_{>0}$, there
exists a set of time instances $\left\{ t_{i}\right\} _{i=1}^{M}$
such that a history stack recorded using (\ref{eq:Online Recording})
is full rank.
\end{assumption}
A singular value maximization algorithm is used to select the time
instances $\left\{ t_{i}\right\} _{i=1}^{M}$. That is, a data-point
$\left(P_{j},\hat{F}_{j},\hat{G}_{j}\right)$ in the history stack
is replaced with a new data-point $\left(P^{*},\hat{F}^{*},\hat{G}^{*}\right)$,
where $\hat{F}^{*}=\hat{F}\left(t\right)$, $P^{*}=P\left(t\right)$,
and $\hat{G}^{*}=\hat{G}\left(t\right)$, for some $t$, only if 
\begin{equation}
\textnormal{s}_{\min}\!\left(\!\sum_{i\neq j}\hat{G}_{i}\hat{G}_{i}^{T}\!+\!\hat{G}_{j}\hat{G}_{j}^{T}\!\right)\!<\frac{\!\textnormal{s}_{\min}\!\left(\!\sum_{i\neq j}\hat{G}_{i}\hat{G}_{i}^{T}\!+\!\hat{G}^{*}\hat{G}^{*T}\!\right)}{\left(1+\zeta\right)},\label{eq:smax}
\end{equation}
where $\textnormal{s}_{\min}\left(\cdot\right)$ denotes the minimum
singular value of a matrix and $\zeta$ is a constant. To simplify
the analysis, new data points are assumed to be collected $\tau_{1}+\tau_{2}$
seconds after a purging event. Since the history stack is updated
using a singular value maximization algorithm, the matrix $\mathscr{G}$
is a piece-wise constant function of time. The use of singular value
maximization to update the history stack implies that once the matrix
$\mathscr{G}$ satisfies (\ref{eq:Rank Condition}), at some $t=T$,
and for some $\underline{c}$, the condition $\underline{c}<\lambda_{\min}\left(\mathscr{G}\left(t\right)\right)$
holds for all $t\geq T$. The developed purging method is summarized
in Fig. \ref{alg:CLNoXDotpurgeDwell}.
\begin{figure}
\begin{algorithmic}[1]

\State$\delta\left(T_{0}\right)\leftarrow0$, $\eta\left(T_{0}\right)\leftarrow0$

\If{$t>\delta\left(t\right)+\tau_{1}+\tau_{2}$ and a data point
is available}

\If{$\mathcal{G}$ is not full}

\State add the data point to $\mathcal{G}$

\Else

\State add the data point to $\mathcal{G}$ if (\ref{eq:smax}) holds

\EndIf

\If{ $\textnormal{s}_{\min}\left(\mathscr{G}\right)\geq\xi\eta\left(t\right)$}

\If{$t-\delta\left(t\right)\geq\mathcal{T}\left(t\right)$ }

\State $\mathcal{H}\leftarrow\mathcal{G}$ and $\mathcal{G}\leftarrow0$
\Comment{purge and replace $\mathcal{H}$}

\State $\delta\left(t\right)\leftarrow t$

\If{$\eta\left(t\right)<\textnormal{s}_{\min}\left(\mathscr{G}\right)$}

\State $\eta\left(t\right)\leftarrow\textnormal{s}_{\min}\left(\mathscr{G}\right)$

\EndIf

\EndIf

\EndIf

\EndIf

\end{algorithmic}

\caption{\label{alg:CLNoXDotpurgeDwell}Algorithm for history stack purging
with dwell time. At each time instance $t$, $\delta\left(t\right)$
stores the last time instance $\mathcal{H}$ was purged, $\eta\left(t\right)$
stores the highest minimum singular value of $\mathscr{G}$ encountered
so far, $\mathcal{T}\left(t\right)$ denotes the dwell time, and $\xi\in\left(0,1\right]$
denotes a threshold fraction.}
\end{figure}

A Lyapunov-based analysis of the parameter and the state estimation
errors is presented in the following section.

\section{\label{sec:Stability-Analysis}Stability Analysis}

Each purging event represents a discontinuous change in the system
dynamics; hence, the resulting closed-loop system is a switched system.
To facilitate the analysis of the switched system, let $\rho:\mathbb{R}_{\geq0}\to\mathbb{N}$
denote a switching signal such that $\rho\left(0\right)=1$, and $\rho\left(t\right)=i+1$,
where $i$ denotes the number of times the update $\mathcal{H}\leftarrow\mathcal{G}$
was carried out over the time interval $\left(0,t\right)$. For some
$s\in\mathbb{N}$, let $\mathcal{H}_{s}$ denotes the history stack
active during the time interval $\left\{ t\mid\rho\left(t\right)=s\right\} $),
containing the elements $\left\{ \left(P_{si},\hat{F}_{si},\hat{G}_{si}\right)\right\} _{i=1,\cdots,M}$,
and let $\mathcal{E}_{si}^{T}$ be the corresponding error term. To
simplify the notation, let $\mathscr{G}_{s}\triangleq\sum_{i=1}^{M}\hat{G}_{si}\hat{G}_{si}^{T}$,
and $Q_{s}=\sum_{i=1}^{M}\hat{G}_{si}\mathcal{E}_{si}^{T}$.

Using (\ref{eq:History Stack Compatibility}) and (\ref{eq:Theta Dynamics}),
the dynamics of the parameter estimation error can be expressed as
\begin{equation}
\dot{\tilde{\theta}}\left(t\right)=-k_{\theta}\Gamma\left(t\right)\mathscr{G}_{s}\left(t\right)\tilde{\theta}\left(t\right)-k_{\theta}\Gamma\left(t\right)Q_{s}\left(t\right).\label{eq:Theta Error Dynamics}
\end{equation}
Since the functions $\mathscr{G}_{s}:\r_{\geq T_{0}}\to\r^{p\times p}$
and $Q_{s}:\mathbb{R}_{\geq T_{0}}\to\mathbb{R}^{p\times n}$ are
piece-wise continuous, the trajectories of (\ref{eq:Theta Error Dynamics}),
and of all the subsequent error systems involving $\mathscr{G}_{s}$
and and $Q_{s}$, are defined in the sense of Carath\'{e}odory. Algorithm
\ref{alg:CLNoXDotpurgeDwell} ensures that there exists a constant
$\underline{g}>0$ such that $\lambda_{\min}\left\{ \mathscr{G}_{s}\right\} \geq\underline{g},\:\forall s\in\mathbb{N}$.

Using the dynamics in (\ref{eq:Linear System}), (\ref{eq:Estimator})
- (\ref{eq:eta Update}), and the design of the feedback component
in (\ref{eq:Feedback}), the time-derivative of the error signal $r$
is given by
\begin{multline}
\dot{r}\left(t\right)=-kr\left(t\right)+\tilde{f^{o}}\left(x,u,\hat{x}\right)+\theta^{T}\tilde{\sigma}\left(x,u,\hat{x}\right)-\tilde{\theta}^{T}\tilde{\sigma}\left(x,u,\hat{x}\right)\\
+\tilde{\theta}^{T}\sigma\left(x,u\right)+\epsilon\left(x,u\right)-\alpha^{2}\tilde{p}+\left(k+\alpha\right)\eta,\label{eq:DYNr}
\end{multline}
where $\tilde{\sigma}\left(x,u,\hat{x}\right)=\sigma\left(x,u\right)-\sigma\left(\hat{x},u\right)$
and $\tilde{f^{o}}\left(x,u,\hat{x}\right)=f\left(x,u\right)-f\left(\hat{x},u\right)$.
Since $\left(x,u\right)\mapsto f\left(x,u\right)$ and $\left(x,u\right)\mapsto\sigma\left(x,u\right)$
are locally Lipschitz, and since $t\mapsto u\left(t\right)$ is bounded,
given a compact set $\hat{\chi}\subset\r^{n}\times\r^{m}\times\r^{n}$,
there exist $L_{f},L_{\sigma}>0$ such that $\sup_{\left(x,u,\hat{x}\right)\in\hat{\chi}}\left\Vert \tilde{f^{o}}\left(x,u,\hat{x}\right)\right\Vert \leq L_{f}\left\Vert \tilde{x}\right\Vert $
and $\sup_{\left(x,u,\hat{x}\right)\in\hat{\chi}}\left\Vert \tilde{\sigma}\left(x,u,\hat{x}\right)\right\Vert \leq L_{\sigma}\left\Vert \tilde{x}\right\Vert $.

To facilitate the analysis, let $\left\{ T_{s}\in\mathbb{R}_{\geq0}\mid s\in\mathbb{N}\right\} $
be a set of switching time instances defined as $T_{s}=\left\{ t\!\mid\!\rho\left(\tau\right)\!<s+1,\forall\tau\in\left[0,t\right)\land\rho\left(\tau\right)\geq s+1,\forall\tau\in\left[t,\infty\right)\right\} .$
That is, for a given switching index $s,$ $T_{s}$ denotes the time
instance when the $\left(s+1\right)$\textsuperscript{th} subsystem
is switched on. The analysis is carried out separately over the time
intervals $\left[T_{s-1},T_{s}\right)$, $s\in\mathbb{N}$, where
$T_{1}=T_{0}+\tau_{1}+\tau_{2}+t_{M}$. Since the history stack $\mathcal{H}$
is not updated over the intervals $\left[T_{s-1},T_{s}\right)$, $s\in\mathbb{N}$,
the matrices $\mathscr{G}_{s}$ and $Q_{s}$ are constant over each
individual interval. The history stack that is active over the interval
$\left[T_{s},T_{s+1}\right)$ is denoted by $\mathcal{H}_{s}$. To
ensure boundedness of the trajectories in the interval $t\in\left[T_{0},T_{1}\right)$,
the history stack $\mathcal{H}_{1}$ is arbitrarily selected to be
full rank. The analysis is carried out over the aforementioned intervals
using the state vectors $Z\triangleq\left[\begin{array}{cccc}
\tilde{p}^{T} & r^{T} & \eta^{T} & \mbox{vec}\left(\tilde{\theta}\right)^{T}\end{array}\right]^{T}\in\mathbb{R}^{3n+np}$ and $Y\triangleq\begin{bmatrix}\tilde{p}^{T} & r^{T} & \eta^{T}\end{bmatrix}^{T}\in\r^{3n}$
as follows. 

\emph{Interval 1:} First, it is established that $Z$ is bounded over
$\left[T_{0},T_{1}\right)$, where the bound is $O\left(\left\Vert Z\left(T_{0}\right)\right\Vert +\left\Vert \sum_{i=1}^{M}\mathcal{E}_{1i}\right\Vert +\overline{\epsilon}\right)$.
Given some $\varepsilon>0$, the bound on $Z$ is utilized to select
gains such that $\left\Vert Y\left(T_{1}\right)\right\Vert <\varepsilon$.

\emph{Interval 2:} The history stack $\mathcal{H}_{2}$, which is
active over $\left[T_{1},T_{2}\right)$, is recorded over $\left[T_{0},T_{1}\right)$.
Without loss of generality, it is assumed that $\mathcal{H}_{2}$
represents the system better than $\mathcal{H}_{1}$ (which is arbitrarily
selected), that is, $\left\Vert \sum_{i=1}^{M}\mathcal{E}_{1i}\right\Vert \geq\left\Vert \sum_{i=1}^{M}\mathcal{E}_{2i}\right\Vert $.
The bound on $Z$ over $\left[T_{1},T_{2}\right)$ is then shown to
be smaller than that over $\left[T_{0},T_{1}\right)$, which utilized
to show that $\left\Vert Y\left(t\right)\right\Vert \leq\varepsilon,$
for all $t\in\left[T_{1},T_{2}\right)$. 

\emph{Interval 3: }Using (\ref{eq:ErrorBound}), the errors $\mathcal{E}_{3i}$
are shown to be $O\left(\left\Vert Y_{3i}\right\Vert +\overline{\epsilon}\right)$
where $Y_{3i}$ denotes the value of $Y$ at the time when the point
$\left(P_{3i},\hat{F}_{3i},\hat{G}_{3i}\right)$ was recorded. Using
the facts that the history stack $\mathcal{H}_{3}$, which is active
over $\left[T_{2},T_{3}\right)$, is recorded over $\left[T_{1},T_{2}\right)$
and $\left\Vert Y\left(t\right)\right\Vert \leq\varepsilon,$ for
all $t\in\left[T_{1},T_{2}\right)$, the error $\left\Vert \sum_{i=1}^{M}\mathcal{E}_{3i}\right\Vert $
is shown to be $O\left(\varepsilon+\overline{\epsilon}\right)$. If
$T_{3}=\infty$ then it is established that $\lim\sup_{t\to\infty}\left\Vert Z\left(t\right)\right\Vert =O\left(\varepsilon+\overline{\epsilon}\right)$.
If $T_{3}<\infty$ then the fact that the bound on $Z$ over $\left[T_{2},T_{3}\right)$
is smaller than that over $\left[T_{1},T_{2}\right)$ is utilized
to show that $\left\Vert Y\left(t\right)\right\Vert \leq\varepsilon,$
for all $t\in\left[T_{2},T_{3}\right)$. The analysis is then continued
in an inductive argument to show that $\lim\sup_{t\to\infty}\left\Vert Z\left(t\right)\right\Vert =O\left(\varepsilon+\overline{\epsilon}\right)$
and $\left\Vert Y\left(t\right)\right\Vert \leq\varepsilon,$ for
all $t\in\left[T_{2},\infty\right)$.

The stability result is summarized in the following theorem.
\begin{thm}
Let $\varepsilon>0$ be given. Let the history stacks $\mathcal{H}$
and $\mathcal{G}$ be populated using the algorithm detailed in Fig.
\ref{alg:CLNoXDotpurgeDwell}. Let the learning gains be selected
to satisfy the sufficient gain conditions in (\ref{eq:V Gain Conditions 1}),
(\ref{eq:V gain condition}), (\ref{eq:W Gain Conditions 1}), and
(\ref{eq:Z0 condition}). Let $T\in\mathbb{R}_{>0}$ be a time instance
such that the system states are exciting over $\left[T_{0},T\right]$,
that is, the history stack can be replenished if purged at any time
$t\in\left[T_{0},T\right]$. Assume that over each switching interval
$\left\{ t\mid\rho\left(t\right)=s\right\} $, the dwell-time, $\mathcal{T}$,
is selected such that $\mathcal{T}\left(t\right)=\mathcal{T}_{s}$,
where%
\begin{comment}
$\mathcal{T}_{s}=T_{s}-T_{s-1}$
\end{comment}
{} $\mathcal{T}_{s}$ is selected to be large enough to satisfy (\ref{eq:Dwell}).
Furthermore assume that the excitation interval is large enough so
that $T_{2}<T$ .\footnote{A minimum of two purges are required to remove the randomly initialized
data, and the data recorded during transient phase of the derivative
estimator from the history stack.} Then, $\lim\sup_{t\to\infty}\left\Vert Z\left(t\right)\right\Vert =O\left(\varepsilon+\overline{\epsilon}\right)$.
\end{thm}
\begin{proof}
Provided $\mathcal{H}_{1}$ is full rank, then the candidate Lyapunov
function 
\begin{equation}
2V\left(Z,t\right)\triangleq\alpha^{2}\tilde{p}^{T}\tilde{p}+r^{T}r+\eta^{T}\eta+\tr\left(\tilde{\theta}^{T}\Gamma^{-1}\left(t\right)\tilde{\theta}\right)\label{eq:V}
\end{equation}
can be utilized to establish boundedness of trajectories over $\left[T_{s-1},T_{s}\right)$.
The candidate Lyapunov function satisfies 
\begin{equation}
\underline{v}\left\Vert Z\right\Vert ^{2}\leq V\left(Z,t\right)\leq\overline{v}\left\Vert Z\right\Vert ^{2},\label{eq:CLNoXDotVBounds}
\end{equation}
where $\overline{v}\triangleq\frac{1}{2}\max\left\{ 1,\alpha^{2},\nicefrac{1}{\underline{\Gamma}}\right\} $
and $\underline{v}\triangleq\frac{1}{2}\min\left\{ 1,\alpha^{2},\nicefrac{1}{\overline{\Gamma}}\right\} $.
The time-derivative of $V$ along the trajectories of (\ref{eq:r}),
(\ref{eq:eta Update}), (\ref{eq:Gamma Dynamics}), (\ref{eq:Theta Error Dynamics}),
and (\ref{eq:DYNr}) is given by
\begin{gather*}
\dot{V}\!=-\alpha^{3}\tilde{p}^{T}\tilde{p}-kr^{T}r-\beta\eta^{T}\eta-\frac{1}{2}\tr\left(\!\tilde{\theta}^{T}\!\!\left(k_{\theta}\mathscr{G}_{1}+\beta_{1}\Gamma^{-1}\right)\!\tilde{\theta}\right)\\
+r^{T}\tilde{f^{o}}+r^{T}\theta^{T}\tilde{\sigma}+r^{T}\tilde{\theta}^{T}\hat{\sigma}+r^{T}\epsilon-k_{\theta}\tr\left(\tilde{\theta}^{T}Q_{s}\right).
\end{gather*}
Using the Cauchy-Schwartz inequality, the derivative can be bounded
as
\begin{gather*}
\dot{V}\!\leq-\alpha^{3}\left\Vert \tilde{p}\right\Vert ^{2}-k\left\Vert r\right\Vert ^{2}-\beta\left\Vert \eta\right\Vert ^{2}-\frac{1}{2}\underline{a}\left\Vert \tilde{\theta}\right\Vert ^{2}+L_{f}\left\Vert r\right\Vert \left\Vert \tilde{x}\right\Vert \\
+\left\Vert r\right\Vert \overline{\theta}L_{\sigma}\left\Vert \tilde{x}\right\Vert +\left\Vert r\right\Vert \left\Vert \tilde{\theta}\right\Vert \left\Vert \hat{\sigma}\right\Vert +\left\Vert r\right\Vert \overline{\epsilon}+k_{\theta}\left\Vert \tilde{\theta}\right\Vert \overline{Q}_{s},
\end{gather*}
where $\underline{a}=k_{\theta}\underline{g}+\frac{\beta_{1}}{\overline{\Gamma}},$
and $\overline{Q}_{s}$ is a positive constant such that $\overline{Q}_{s}\geq\left\Vert Q_{s}\right\Vert $.
Provided 
\begin{align}
k> & \max\left(2\left(4+\alpha\right)\left(L_{f}+\overline{\theta}L_{\sigma}\right),\frac{12\overline{\sigma}^{2}}{\underline{a}}\right),\nonumber \\
\alpha^{3}> & \left(1+\alpha\right)\left(L_{f}+\overline{\theta}L_{\sigma}\right),\nonumber \\
\beta> & \left(L_{f}+\overline{\theta}L_{\sigma}\right),\label{eq:V Gain Conditions 1}
\end{align}
Young's inequality and nonlinear damping can be used to conclude that
\begin{gather*}
\dot{V}\!\leq-\frac{\alpha^{3}}{2}\left\Vert \tilde{p}\right\Vert ^{2}-\frac{k}{4}\left\Vert r\right\Vert ^{2}-\frac{\beta}{2}\left\Vert \eta\right\Vert ^{2}-\frac{\underline{a}}{6}\left\Vert \tilde{\theta}\right\Vert ^{2}\\
-\left(\frac{k}{8}-\frac{3L_{\sigma}}{2\underline{a}}\left\Vert \tilde{x}\right\Vert ^{2}\right)\left\Vert r\right\Vert ^{2}+\frac{\overline{\epsilon}^{2}}{k}+\frac{3k_{\theta}^{2}}{2\underline{a}}\overline{Q}_{s}^{2},
\end{gather*}
 Since $\left\Vert \tilde{x}\right\Vert ^{2}\leq\left(1+\alpha\right)\left\Vert Z\right\Vert ^{2}$,
$\dot{V}\leq-\nu\left(\left\Vert Z\right\Vert -\frac{\iota_{s}}{\nu}\right)$,
in the domain 
\[
\mathcal{D}\triangleq\left\{ Z\in\r^{3n+np}\mid\left\Vert Z\right\Vert <\sqrt{\frac{k\underline{a}}{12L_{\sigma}\left(1+\alpha\right)}}\right\} .
\]
That is, $\dot{V}$ is negative definite on $\mathcal{D}$ provided
$\left\Vert Z\right\Vert >\sqrt{\frac{\iota_{s}}{\nu}}>0$, where
$v\triangleq\frac{1}{2}\min\left\{ \alpha^{3},\nicefrac{k}{2},\beta,\nicefrac{\underline{a}}{3}\right\} $
and $\iota_{s}\triangleq\frac{\overline{\epsilon}^{2}}{k}+\frac{3k_{\theta}^{2}}{2\underline{a}}\overline{Q}_{s}^{2}$.
Theorem 4.18 from \cite{SCC.Khalil2002} can then be invoked to conclude
that provided 
\begin{equation}
k>\frac{15L_{\sigma}\left(1+\alpha\right)}{\underline{a}\underline{v}}\max\left(\overline{V}_{s},\frac{\overline{v}\iota_{1}}{v}\right),\label{eq:V gain condition}
\end{equation}
where $\overline{V}_{s}\geq\left\Vert V\left(Z\left(T_{s-1}\right),T_{s-1}\right)\right\Vert $
is a constant, then $\dot{V}\leq-\frac{v}{\overline{v}}V+\iota_{s},$
$\forall t\in\left[T_{s-1},T_{s}\right)$.

In particular, $\forall t\in\left[T_{0},T_{1}\right),$
\begin{equation}
V\left(Z\left(t\right),t\right)\leq\left(\overline{V}_{1}-\frac{\overline{v}}{v}\iota_{1}\right)\e^{-\frac{v}{\overline{v}}\left(t-T_{0}\right)}+\frac{\overline{v}}{v}\iota_{1},\label{eq:CLNoXDotVDecay}
\end{equation}
where $\overline{V}_{1}>0$ is a constant such that $\left|V\left(Z\left(T_{0}\right),T_{0}\right)\right|\leq\overline{V}_{1}$.
Hence, $\forall t\in\left[T_{0},T_{1}\right),$
\begin{equation}
\left\Vert \tilde{\theta}\left(t\right)\right\Vert \leq\theta_{1}\triangleq\sqrt{\frac{1}{\underline{v}}}\max\left\{ \sqrt{\overline{V}_{1}},\sqrt{\frac{\overline{v}}{v}\iota_{1}}\right\} .\label{eq:CLNoXDotThetaTildeBound}
\end{equation}

If it were possible to use the inequality in (\ref{eq:CLNoXDotVDecay})
to conclude that over $\left[T_{0},T_{1}\right)$, $V\left(Z\left(t\right),t\right)\leq V\left(Z\left(T_{0}\right),T_{0}\right)$,
then an inductive argument could be used to show that the trajectories
decay to a neighborhood of the origin. However, unless the history
stack can be selected to have arbitrarily large minimum singular value
(which is generally not possible), the constant $\frac{\overline{v}}{v}\iota_{1}$
cannot be made arbitrarily small using the learning gains.

Since $\iota_{s}$ depends on $Q_{s}$, it can be made smaller by
reducing the estimation errors and thereby reducing the errors associated
with the data stored in the history stack. To that end, consider the
candidate Lyapunov function 
\begin{equation}
W\left(Y\right)\triangleq\frac{\alpha^{2}}{2}\tilde{p}^{T}\tilde{p}+\frac{1}{2}r^{T}r+\frac{1}{2}\eta^{T}\eta.\label{eq:W}
\end{equation}
The candidate Lyapunov function satisfies 
\begin{equation}
\underline{w}\left\Vert Y\right\Vert ^{2}\leq W\left(Y,t\right)\leq\overline{w}\left\Vert Y\right\Vert ^{2},\label{eq:CLNoXDotVBounds-1}
\end{equation}
where $\overline{w}\triangleq\frac{1}{2}\max\left\{ 1,\alpha^{2}\right\} $,
$\underline{w}\triangleq\frac{1}{2}\min\left\{ 1,\alpha^{2}\right\} $.
In the interval $\left[T_{s-1},T_{s}\right)$, the time-derivative
of $W$ is given by 
\begin{multline*}
\dot{W}=-\alpha^{3}\tilde{p}^{T}\tilde{p}-kr^{T}r-\beta\eta^{T}\eta+r^{T}\left(\tilde{f^{o}}+\left(\theta^{T}-\tilde{\theta}^{T}\right)\tilde{\sigma}\right)\\
+r^{T}\left(\tilde{\theta}^{T}\sigma+\epsilon\right)
\end{multline*}
Using the Cauchy-Schwartz inequality, the derivative $\dot{W}$ can
be bounded as
\begin{align*}
\dot{W}= & -\alpha^{3}\left\Vert \tilde{p}\right\Vert ^{2}-k\left\Vert r\right\Vert ^{2}-\beta\left\Vert \eta\right\Vert ^{2}\\
 & +\left(L_{f}+\left(\overline{\theta}+\theta_{s}\right)L_{\sigma}\right)\left\Vert r\right\Vert \left\Vert \tilde{x}\right\Vert +\left(\theta_{s}\overline{\sigma}+\overline{\epsilon}\right)\left\Vert r\right\Vert ,
\end{align*}
where $\theta_{s}>0$ is a constant such that $\theta_{s}\geq\sup_{t\in\left[T_{s-1},T_{s}\right)}\left\Vert \tilde{\theta}\left(t\right)\right\Vert $.

Consider the time interval $\left[T_{0},T_{1}\right)$. Provided
\begin{align}
k & \geq1+\theta_{1}^{2}+\left(L_{f}+\left(\overline{\theta}+\theta_{1}\right)L_{\sigma}\right)\left(4+\alpha\right)\nonumber \\
\alpha^{3} & \geq\left(1+\alpha\right)\left(L_{f}+\left(\overline{\theta}+\theta_{1}\right)L_{\sigma}\right)\nonumber \\
\beta & \geq\left(L_{f}+\left(\overline{\theta}+\theta_{1}\right)L_{\sigma}\right)\label{eq:W Gain Conditions 1}
\end{align}
then $\dot{W}\leq-\frac{w}{\overline{w}}W+\lambda,$ where $w=\frac{1}{2}\min\left(\alpha^{3},k,\beta\right)$
and $\lambda=\frac{\overline{\sigma}^{2}+\overline{\epsilon}^{2}}{2}$.
That is, for all $t\in\left[T_{0},T_{1}\right),$
\begin{equation}
W\left(Y\left(t\right),t\right)\leq\left(\overline{W}_{1}-\frac{\overline{w}}{w}\lambda\right)\e^{-\frac{w}{\overline{w}}\left(t-T_{0}\right)}+\frac{\overline{w}}{w}\lambda,\label{eq:CLNoXDotVDecay-1}
\end{equation}
where $\overline{W}_{1}>0$ is a constant such that $\left|W\left(Y\left(T_{0}\right)\right)\right|\leq\overline{W}_{1}$.
In particular, $\forall t\in\left[T_{0},T_{1}\right).$ 
\begin{equation}
\left\Vert Y\left(t\right)\right\Vert \leq\sqrt{\frac{1}{\underline{w}}\max\left(\overline{W}_{1},\frac{\overline{w}}{w}\lambda\right)}\triangleq\overline{\left\Vert Y\right\Vert }_{1}.\label{eq:YBound}
\end{equation}
Provided the dwell time $\mathcal{T}_{s}$ is large enough so that
\begin{align}
\left(\overline{W}_{s}-\frac{\overline{w}}{w}\lambda\right)e^{-\frac{w}{\overline{w}}\mathcal{T}_{s}} & \leq\frac{\overline{w}}{w}\lambda,\nonumber \\
\left(\overline{V}_{s}-\frac{\overline{v}}{v}\iota_{s}\right)e^{-\frac{v}{\overline{v}}\mathcal{T}_{s}} & \leq\frac{\overline{v}}{v}\iota_{s},\label{eq:Dwell}
\end{align}
then from (\ref{eq:CLNoXDotVDecay}) and (\ref{eq:CLNoXDotVDecay-1}),
$W\left(Y\left(T_{1}\right)\right)\leq\frac{2\overline{w}\lambda}{w}$
and $V\left(Z\left(T_{1}\right),T_{1}\right)\leq\frac{2\overline{v}\iota_{1}}{v}$.
In particular, $\left\Vert Y\left(T_{1}\right)\right\Vert \leq\sqrt{\frac{2\overline{w}\lambda}{\underline{w}w}}$
and $\left\Vert Z\left(T_{1}\right)\right\Vert \leq\sqrt{\frac{2\overline{v}\iota_{1}}{\underline{v}v}}$.
Note that the bound on $Y\left(T_{1}\right)$ can be made arbitrarily
small by increasing $k,$ $\alpha,$ and $\beta$.

Now the interval $\left[T_{1},T_{2}\right)$ is considered. Since
the history stack $\mathcal{H}_{2}$ which is active during $\left[T_{1},T_{2}\right)$
is recorded during $\left[T_{0},T_{1}\right)$, the bound in (\ref{eq:ErrorBound})
can be used to show that $\overline{Q}_{2}=O\left(\overline{\left\Vert Y\right\Vert }_{1}+\overline{\epsilon}\right)$.

Since $\mathcal{H}_{1}$ is independent of the system trajectories,
$\overline{Q}_{1}$ can be selected such that $\overline{Q}_{2}<\overline{Q}_{1}$,
and hence, $\iota_{2}<\iota_{1}$. Thus, provided the constant $\overline{V}_{1}$
(and as a result, the gain $k$) is selected large enough so that
\begin{equation}
\frac{2\overline{v}\iota_{1}}{v}<\overline{V}_{1},\label{eq:Z0 condition}
\end{equation}
the gain condition in (\ref{eq:V gain condition}) holds over $\left[T_{1},T_{2}\right)$,
and hence, a similar Lyapunov-based analysis, along with the bound
$\overline{V}_{2}=\frac{2\overline{v}\iota_{1}}{v}$ can be utilized
to conclude that $\forall t\in\left[T_{1},T_{2}\right)$,
\begin{equation}
\left\Vert \tilde{\theta}\left(t\right)\right\Vert \leq\sqrt{\frac{\overline{v}}{\underline{v}v}}\max\left\{ \sqrt{2\iota_{1}},\sqrt{\iota_{2}}\right\} \triangleq\theta_{2}.\label{eq:CLNoXDotThetaTildeBound-1}
\end{equation}
The sufficient condition in (\ref{eq:Z0 condition}) implies that
$\overline{V}_{2}<\overline{V}_{1}$ and hence, (\ref{eq:CLNoXDotThetaTildeBound})
and $\iota_{2}<\iota_{1}$ imply that $\theta_{2}<\theta_{1}$.

Since $\theta_{2}<\theta_{1}$, the gain conditions in (\ref{eq:W Gain Conditions 1})
hold over the interval $\left[T_{1},T_{2}\right)$. A Lyapunov-based
analysis similar to (\ref{eq:W})-(\ref{eq:YBound}) yields $\left\Vert Y\left(t\right)\right\Vert \leq\sqrt{\frac{1}{\underline{w}}\max\left(\overline{W}_{2},\frac{\overline{w}}{w}\lambda\right)}.$
From (\ref{eq:Dwell}), $\overline{W}_{2}=\frac{2\overline{w}\lambda}{w}$,
and hence, $\forall t\in\left[T_{1},T_{2}\right)$,
\begin{equation}
\left\Vert Y\left(t\right)\right\Vert \leq\sqrt{\frac{2\overline{w}\lambda}{\underline{w}w}}\triangleq\overline{\left\Vert Y\right\Vert }_{2}.\label{eq:YBound-1}
\end{equation}

Now, the interval $\left[T_{2},T_{3}\right)$ is considered. Since
the history stack $\mathcal{H}_{3}$ which is active during $\left[T_{2},T_{3}\right)$
is recorded during $\left[T_{1},T_{2}\right)$, the bounds in (\ref{eq:ErrorBound})
and (\ref{eq:YBound-1}) can be used to show that $\overline{Q}_{3}=O\left(\overline{\left\Vert Y\right\Vert }_{2}+\overline{\epsilon}\right).$
By selecting $\overline{W}_{1}$ large enough, it can be ensured that
$\overline{\left\Vert Y\right\Vert }_{2}<\overline{\left\Vert Y\right\Vert }_{1}$,
and hence, $\overline{Q}_{3}<\overline{Q}_{2}$, which implies $\iota_{3}<\iota_{2}$.
Provided $\mathcal{T}_{2}$ satisfies (\ref{eq:Dwell}), then $\left(\overline{V}_{2}-\frac{\overline{v}}{v}\iota_{2}\right)e^{-\frac{v}{\overline{v}}\left(T_{2}-T_{1}\right)}\leq\frac{\overline{v}}{v}\iota_{2}$,
which implies $\overline{V}_{3}=\frac{\overline{v}}{v}\iota_{2}$,
and hence, $\overline{V}_{3}<\overline{V}_{2}$ and $\theta_{3}<\theta_{2}$.
Therefore, the gain conditions in (\ref{eq:V Gain Conditions 1}),
(\ref{eq:V gain condition}), and (\ref{eq:W Gain Conditions 1})
are satisfied over $\left[T_{2},T_{3}\right)$. 

Since the gain conditions are satisfied, a Lyapunov-based analysis
similar to (\ref{eq:W})-(\ref{eq:YBound}) yields $\left\Vert Y\left(t\right)\right\Vert \leq\sqrt{\frac{2\overline{w}\lambda}{\underline{w}w}},\forall t\in\left[T_{2},T_{3}\right)$.
Given any $\varepsilon>0,$ the gains $\alpha,$ $\beta,$ and $k$
can be selected large enough to satisfy $\overline{\left\Vert Y\right\Vert }_{2}\leq\varepsilon,$
and hence, $\left\Vert Y\left(t\right)\right\Vert \leq\varepsilon,\forall t\in\left[T_{2},T_{3}\right).$
Furthermore, a similar Lyapunov-based analysis as (\ref{eq:V}) -
(\ref{eq:CLNoXDotVDecay}) yields $V\left(Z\left(t\right),t\right)\leq\left(\overline{V}_{3}-\frac{\overline{v}}{v}\iota_{3}\right)\e^{-\frac{v}{\overline{v}}\left(t-T_{2}\right)}+\frac{\overline{v}}{v}\iota_{3},\forall t\in\left[T_{2},T_{3}\right)$.
If $T_{3}=\infty$ then $\lim\sup_{t\to\infty}V\left(Z\left(t\right),t\right)\leq\frac{\overline{v}}{v}\iota_{3}$,
which, from $\overline{Q}_{3}=O\left(\overline{\left\Vert Y\right\Vert }_{2}+\overline{\epsilon}\right)$
and $\iota_{3}=\frac{\overline{\epsilon}^{2}}{k}+\frac{3k_{\theta}^{2}}{2\underline{a}}\overline{Q}_{3}^{2}$
implies that $\lim\sup_{t\to\infty}\left\Vert Z\left(t\right)\right\Vert =O\left(\varepsilon+\overline{\epsilon}\right)$.

If $T_{3}\neq\infty$ then an inductive continuation of the Lyapunov-based
analysis to the time intervals $\left[T_{s-1},T_{s}\right)$ shows
that provided the dwell time $\mathcal{T}_{s}$ satisfies (\ref{eq:Dwell}),
the gain conditions in (\ref{eq:V Gain Conditions 1}), (\ref{eq:V gain condition}),
and (\ref{eq:W Gain Conditions 1}) are satisfied for all $t>T_{3}$,
the state $Y$ satisfies 
\begin{equation}
\left\Vert Y\left(t\right)\right\Vert \leq\varepsilon,\forall t>T_{1},\label{eq:Y}
\end{equation}
and $Q_{s}\leq Q_{s-1}$, $\iota_{s}\leq\iota_{s-1}$, $\overline{V}_{s}\leq\overline{V}_{s-1}$,
and $\overline{\theta}_{s}\leq\overline{\theta}_{s-1}$, for all $s>3$. 

The bound in (\ref{eq:Y}) and the fact that $\overline{Q}_{s}=O\left(\overline{\left\Vert Y\right\Vert }_{s-1}+\overline{\epsilon}\right)$
indicate that $\overline{Q}_{s}=O\left(\varepsilon+\overline{\epsilon}\right),\forall s\in\mathbb{N}$.
Furthermore, $V\left(Z\left(t\right),t\right)\leq\left(\overline{V}_{s}-\frac{\overline{v}}{v}\iota_{s}\right)\e^{-\frac{v}{\overline{v}}\left(t-T_{s-1}\right)}+\frac{\overline{v}}{v}\iota_{s}$,
$\forall t\in\left[T_{s-1},T_{s}\right)$, $\forall s\in\mathbb{N}$,
which, along with the dwell time requirement, implies that $\lim\sup_{t\to\infty}V\left(Z\left(t\right),t\right)\leq\frac{\overline{v}}{v}\iota_{s}$,
and hence, $\lim\sup_{t\to\infty}\left\Vert Z\left(t\right)\right\Vert =O\left(\varepsilon+\overline{\epsilon}\right)$.
\end{proof}

\section{\label{sec:Simulation}Simulation}

\begin{table}
\caption{\label{tab:Simulation-parameters-for}Simulation parameters for the
different simulation runs. The parameters are selected using trial
and error.}

\begin{tabular}{>{\centering}p{0.35\columnwidth}>{\centering}p{0.25\columnwidth}>{\centering}p{0.25\columnwidth}}
\toprule 
 & \multicolumn{2}{c}{Noise Variance}\tabularnewline
Parameter & 0 & 0.001\tabularnewline
\midrule
$T_{1}$ & 0.5 & 0.9\tabularnewline
$T_{2}$ & 0.3 & 0.5\tabularnewline
$N$ & 50 & 150\tabularnewline
$\Gamma\left(t_{0}\right)$ & $\id_{4}$ & $\id_{4}$\tabularnewline
$\beta_{1}$ & 0.5 & 0.5\tabularnewline
$\alpha$ & 2 & 2\tabularnewline
$k$ & 10 & 10\tabularnewline
$\beta$ & 2 & 2\tabularnewline
$\zeta$ & 0 & 0\tabularnewline
$\xi$ & 0.95 & 0.95\tabularnewline
$k_{\theta}$ & \nicefrac{0.5}{N} & \nicefrac{0.5}{N}\tabularnewline
\bottomrule
\end{tabular}
\end{table}
\begin{figure}
\includegraphics[width=0.9\columnwidth]{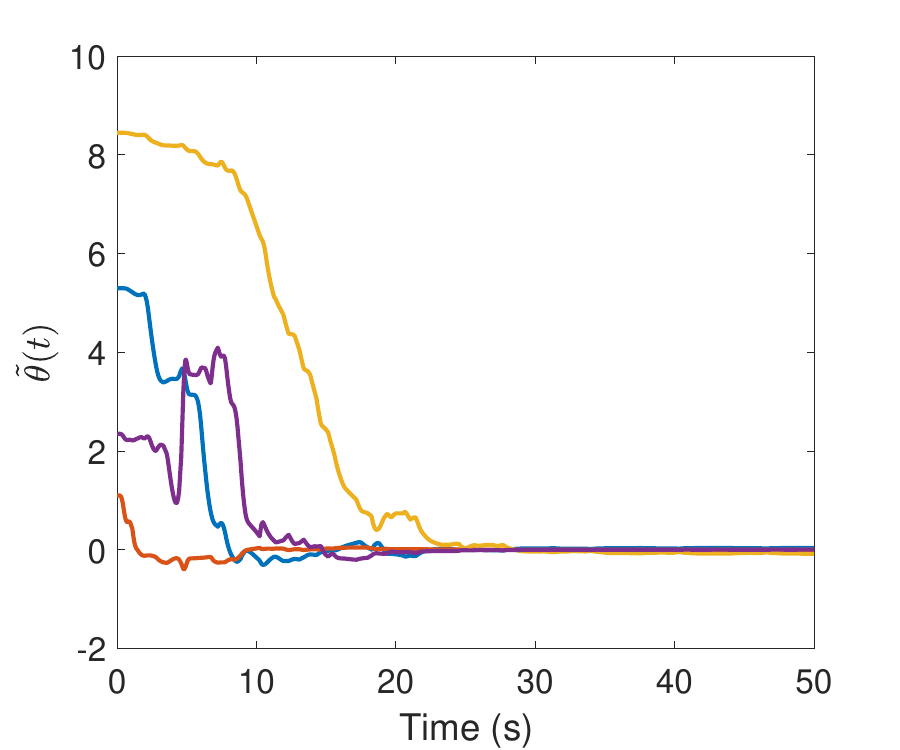}

\caption{\label{fig:ThetaTildeNoNoise}Trajectories of the parameter estimation
errors using noise-free position measurements.}
\end{figure}
\begin{figure}
\includegraphics[width=0.9\columnwidth]{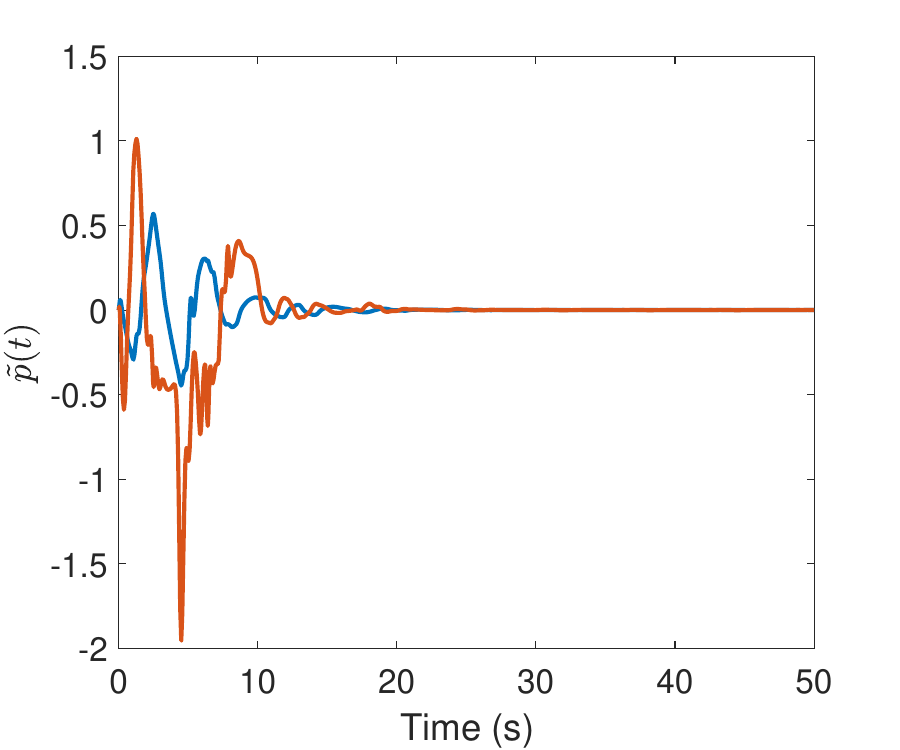}

\caption{\label{fig:pTildeNoNoise}Trajectories of the generalized position
estimation errors using noise-free position measurements.}
\end{figure}
\begin{figure}
\includegraphics[width=0.9\columnwidth]{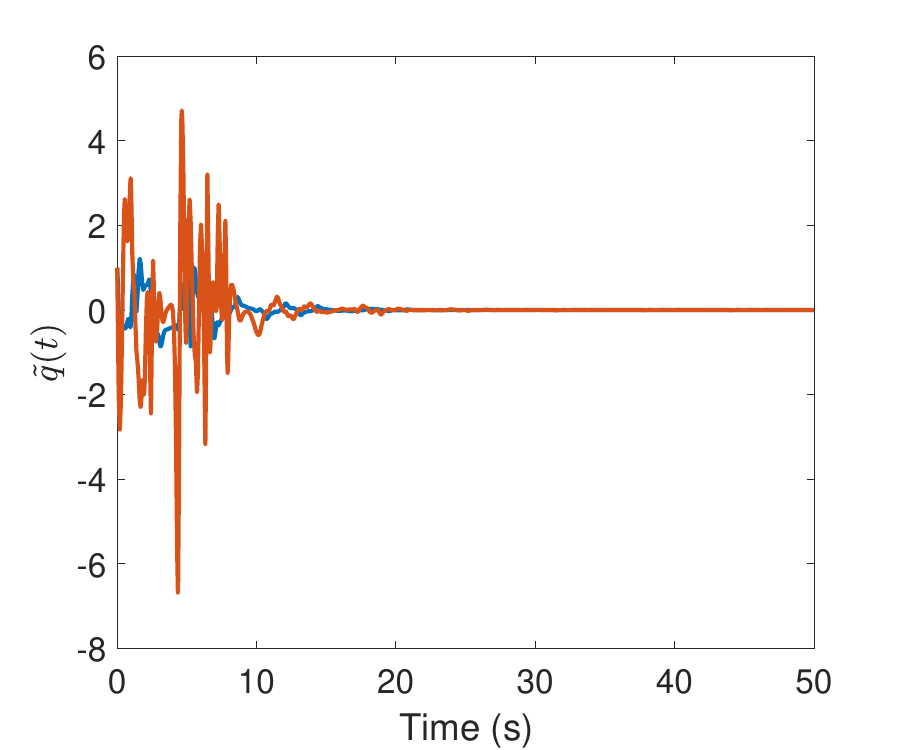}

\caption{\label{fig:qTildeNoNoise}Trajectories of the generalized velocity
estimation errors using noise-free position measurements.}
\end{figure}
\begin{figure}
\includegraphics[width=0.9\columnwidth]{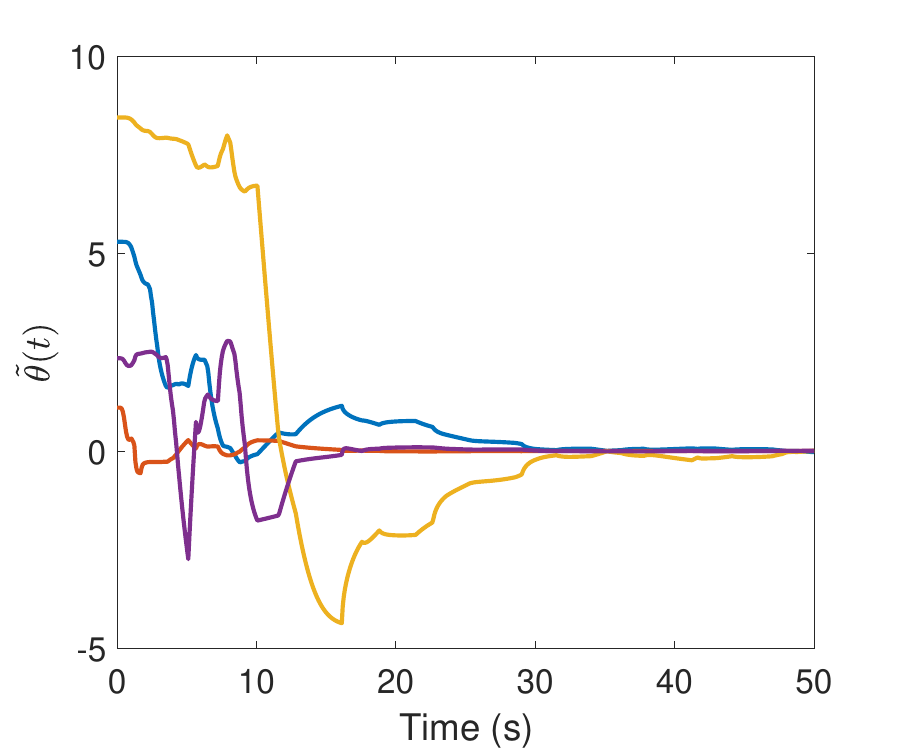}

\caption{\label{fig:ThetaTilde0.001}Trajectories of the parameter estimation
errors with a Gaussian measurement noise (variance = 0.001).}
\end{figure}
\begin{figure}
\includegraphics[width=0.9\columnwidth]{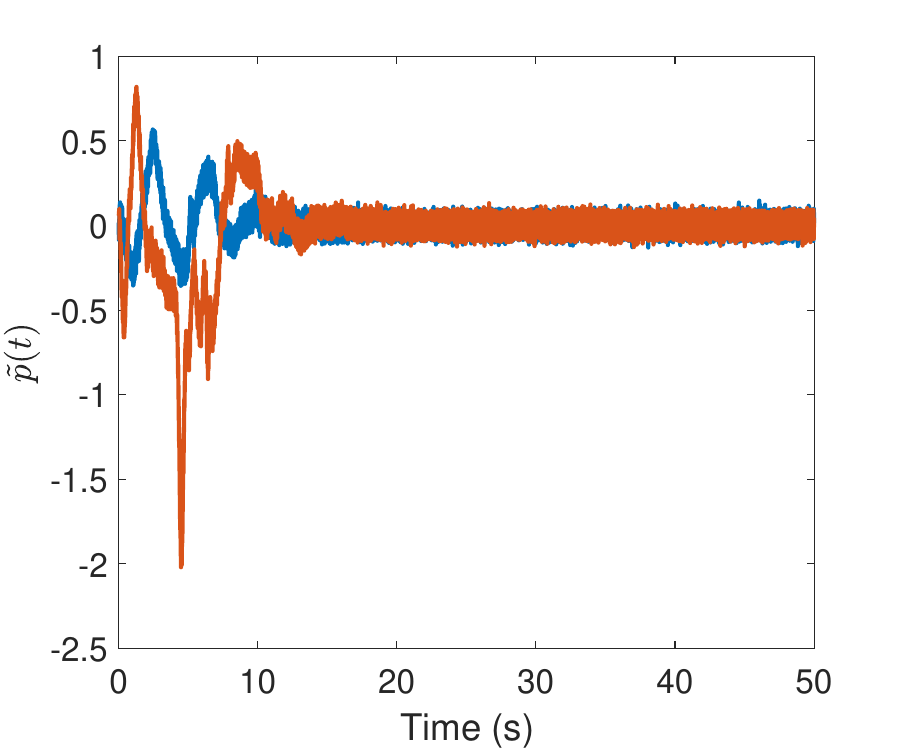}

\caption{\label{fig:pTilde0.001}Trajectories of the generalized position estimation
errors with a Gaussian measurement noise (variance = 0.001).}
\end{figure}
\begin{figure}
\includegraphics[width=0.9\columnwidth]{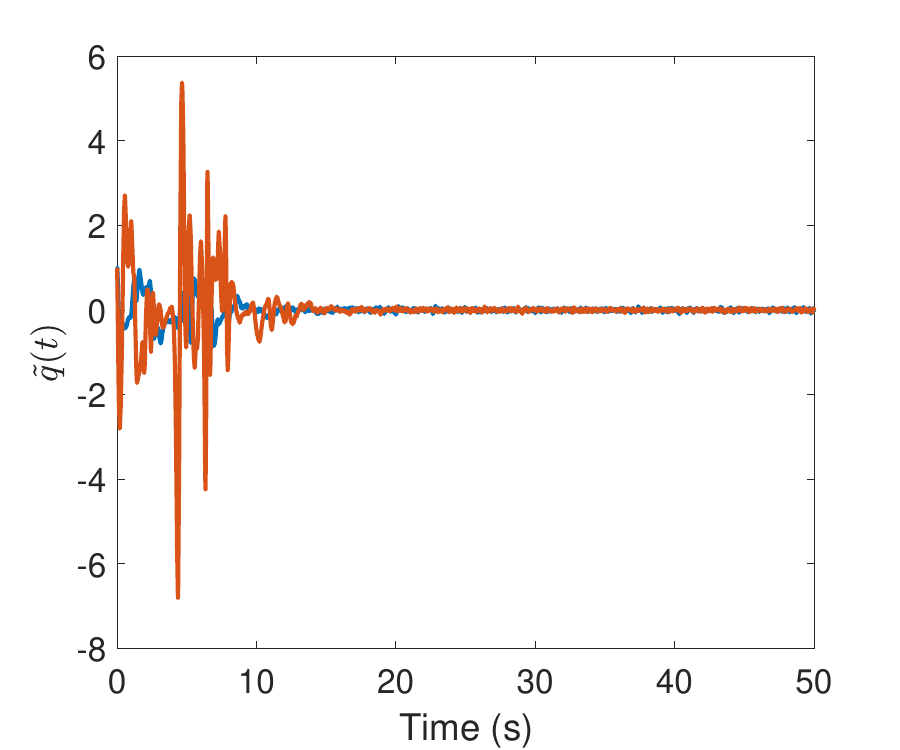}

\caption{\label{fig:qTilde0.001}Trajectories of the generalized velocity estimation
errors with a Gaussian measurement noise (variance = 0.001).}
\end{figure}

The developed technique is simulated using a model for a two-link
robot manipulator arm. The uncertainty $g\left(x,u\right)$ is linearly
parameterizable as $g^{T}\left(x,u\right)=\theta^{T}\sigma\left(x,u\right).$
That is, the selected model belongs to a sub-class of systems defined
by (\ref{eq:Linear System}), where the function approximation error,
$\varepsilon$, is zero. Since the ideal parameters, $\theta$, are
uniquely known, the selected model facilitates quantitative analysis
of the parameter estimation error. The dynamics of the arm are described
by (\ref{eq:Linear System}), where
\begin{gather}
f^{0}\left(x,u\right)=-\left(M\left(p\right)\right)^{-1}V_{m}\left(p,q\right)q+\left(M\left(p\right)\right)^{-1}u,\nonumber \\
g^{T}\left(x,u\right)=\theta^{T}\begin{bmatrix}\left[\begin{array}{cc}
\left(M\left(p\right)\right)^{-1} & \left(M\left(p\right)\right)^{-1}\end{array}\right]D\left(q\right)\end{bmatrix}^{T}.\label{eq:CLNoXDotSymDyn}
\end{gather}
In (\ref{eq:CLNoXDotSymDyn}), $u\in\mathbb{R}^{2}$ is the control
input, $D\left(q\right)\triangleq\mbox{diag}\left[\tanh\left(q_{1}\right),\:\tanh\left(q_{2}\right)\right]$,
$M\left(p\right)\triangleq\begin{bmatrix}p_{1}+2a_{3}\textnormal{c}_{2}\left(p\right), & a_{2}+a_{3}\textnormal{c}_{2}\left(p\right)\\
a_{2}+a_{3}\textnormal{c}_{2}\left(p\right), & a_{2}
\end{bmatrix},$ and $V_{m}\left(p,q\right)\triangleq\begin{bmatrix}-a_{3}\textnormal{s}_{2}\left(p\right)q_{2}, & -a_{3}\textnormal{s}_{2}\left(p\right)\left(q_{1}+q_{2}\right)\\
a_{3}\textnormal{s}_{2}\left(p\right)q_{1}, & 0
\end{bmatrix},$ where $\textnormal{c}_{2}\left(p\right)=\cos\left(p_{2}\right),$
$\textnormal{s}_{2}\left(p\right)=\sin\left(p_{2}\right)$, and $a_{1}=3.473$,
$a_{2}=0.196$, and $a_{3}=0.242$ are constants. The system has four
unknown parameters. The ideal values of the unknown parameters are
$\theta=\begin{bmatrix}5.3 & 1.1 & 8.45 & 2.35\end{bmatrix}^{T}.$ 

The contribution of this paper is the design of a parameter estimator
and a velocity observer. The controller is assumed to be any controller
that results in bounded system response. In this simulation study,
the controller, $u$, is designed so that the system tracks the trajectory
$p_{1}\left(t\right)=p_{2}\left(t\right)=\sin\left(3t\right)+\sin\left(2t\right)$. 

The simulation is performed using Euler forward numerical integration
using a sample time of $T_{s}=0.0005$ seconds. Past $\frac{\tau_{1}+\tau_{2}}{T_{s}}$
values of the generalized position, $p$, and the control input, $u$,
are stored in a buffer. The matrices $P$, $\hat{G}$, and $\hat{F}$
for the parameter update law in (\ref{eq:Theta Dynamics}) are computed
using trapezoidal integration of the data stored in the aforementioned
buffer. Values of $P$, $\hat{G}$, and $\hat{F}$ are stored in the
history stack and are updated according to the algorithm detailed
in Fig. \ref{alg:CLNoXDotpurgeDwell}.

The initial estimates of the unknown parameters are selected to be
zero, and the history stack is initialized so that all the elements
of the history stack are zero. Data is added to the history stack
using a singular value maximization algorithm. To demonstrate the
utility of the developed method, three simulation runs are performed.
In the first run, the observer is assumed to have access to noise
free measurements of the generalized position. In the second run,
a zero-mean Gaussian noise with variance 0.001 is added to the generalized
position signal to simulate measurement noise. The values of various
simulation parameters selected for the three runs are provided in
Table \ref{tab:Simulation-parameters-for}. Figure \ref{fig:ThetaTildeNoNoise}
demonstrates that in absence of noise, the developed parameter estimator
drives the state estimation error, $\tilde{x}$, and the parameter
estimation error, $\tilde{\theta}$, close to the origin. Figures
\ref{fig:ThetaTilde0.001} - \ref{fig:qTilde0.001} indicate that
the developed technique can be utilized in the presence of measurement
noise, with expected degradation of performance.

\section{\label{sec:Conclusion}Conclusion}

This paper develops a concurrent learning based adaptive observer
and parameter estimator to simultaneously estimate the unknown parameters
and the generalized velocity of second-order nonlinear systems using
generalized position measurements. The developed technique utilizes
a dynamic velocity observer to generate state estimates necessary
for data-driven adaptation. A purging algorithm is developed to improve
the quality of the stored data as the state estimates converge to
the true state. By integrating $n-$times, the developed method can
be generalized to higher-order linear systems.

Simulation results indicate that the developed method is robust to
measurement noise. A theoretical analysis of the developed method
under measurement noise and process noise is a subject for future
research. Future efforts will also focus on the examination the effect
of the integration intervals, $\tau_{1}$ and $\tau_{2}$, on the
performance of the developed estimator. \bibliographystyle{IEEEtran}
\bibliography{encr,sccmaster,scc}

\begin{thebibliography}{10}
\def\url#1{}
\csname url@samestyle\endcsname
\providecommand{\newblock}{\relax}
\providecommand{\bibinfo}[2]{#2}
\providecommand{\BIBentrySTDinterwordspacing}{\spaceskip=0pt\relax}
\providecommand{\BIBentryALTinterwordstretchfactor}{4}
\providecommand{\BIBentryALTinterwordspacing}{\spaceskip=\fontdimen2\font plus
\BIBentryALTinterwordstretchfactor\fontdimen3\font minus
  \fontdimen4\font\relax}
\providecommand{\BIBforeignlanguage}[2]{{%
\expandafter\ifx\csname l@#1\endcsname\relax
\typeout{** WARNING: IEEEtran.bst: No hyphenation pattern has been}%
\typeout{** loaded for the language `#1'. Using the pattern for}%
\typeout{** the default language instead.}%
\else
\language=\csname l@#1\endcsname
\fi
#2}}
\providecommand{\BIBdecl}{\relax}
\BIBdecl

\bibitem{SCC.Ioannou.Sun1996}
P.~Ioannou and J.~Sun, \emph{Robust Adaptive Control}.\hskip 1em plus 0.5em
  minus 0.4em\relax Prentice Hall, 1996.

\bibitem{SCC.Sastry.Bodson1989}
S.~Sastry and M.~Bodson, \emph{Adaptive Control: Stability, Convergence, and
  Robustness}.\hskip 1em plus 0.5em minus 0.4em\relax Upper Saddle River, NJ:
  Prentice-Hall, 1989.

\bibitem{SCC.Krstic.Kanellakopoulos.ea1995}
M.~Krstic, I.~Kanellakopoulos, and P.~V. Kokotovic, \emph{Nonlinear and
  Adaptive Control Design}.\hskip 1em plus 0.5em minus 0.4em\relax New York,
  NY, USA: John Wiley \& Sons, 1995.

\bibitem{SCC.Duarte.Narendra1989}
M.~A. Duarte and K.~Narendra, ``Combined direct and indirect approach to
  adaptive control,'' \emph{IEEE Trans. Autom. Control}, vol.~34, no.~10, pp.
  1071--1075, Oct 1989.

\bibitem{SCC.Krstic.Kokotovic.ea1993}
M.~Krsti\'{c}, P.~V. Kokotovi\'{c}, and I.~Kanellakopoulos,
  ``Transient-performance improvement with a new class of adaptive
  controllers,'' \emph{Syst. Control Lett.}, vol.~21, no.~6, pp. 451--461,
  1993.

\bibitem{SCC.Chowdhary.Johnson2011a}
G.~Chowdhary and E.~Johnson, ``A singular value maximizing data recording
  algorithm for concurrent learning,'' in \emph{Proc. Am. Control Conf.}, 2011,
  pp. 3547--3552.

\bibitem{SCC.Chowdhary.Yucelen.ea2013}
G.~Chowdhary, T.~Yucelen, M.~M\"{u}hlegg, and E.~N. Johnson, ``Concurrent
  learning adaptive control of linear systems with exponentially convergent
  bounds,'' \emph{Int. J. Adapt. Control Signal Process.}, vol.~27, no.~4, pp.
  280--301, 2013.

\bibitem{SCC.Kersting.Buss2014}
S.~Kersting and M.~Buss, ``Concurrent learning adaptive identification of
  piecewise affine systems,'' in \emph{Proc. IEEE Conf. Decis. Control}, Dec.
  2014, pp. 3930--3935.

\bibitem{SCC.Chowdhary.Muehlegg.ea2013}
G.~Chowdhary, M.~M\"{u}hlegg, J.~How, and F.~Holzapfel,
  ``\BIBforeignlanguage{English}{Concurrent learning adaptive model predictive
  control},'' in \emph{\BIBforeignlanguage{English}{Advances in Aerospace
  Guidance, Navigation and Control}}, Q.~Chu, B.~Mulder, D.~Choukroun, E.-J.
  van Kampen, C.~de~Visser, and G.~Looye, Eds.\hskip 1em plus 0.5em minus
  0.4em\relax Springer Berlin Heidelberg, 2013, pp. 29--47.

\bibitem{SCC.Modares.Lewis.ea2014}
H.~Modares, F.~L. Lewis, and M.-B. Naghibi-Sistani, ``Integral reinforcement
  learning and experience replay for adaptive optimal control of
  partially-unknown constrained-input continuous-time systems,''
  \emph{Automatica}, vol.~50, no.~1, pp. 193--202, 2014.

\bibitem{SCC.Kamalapurkar.Klotz.ea2014a}
\BIBentryALTinterwordspacing
R.~Kamalapurkar, J.~Klotz, and W.~E. Dixon, ``Concurrent learning-based online
  approximate feedback {N}ash equilibrium solution of {$N$}-player nonzero-sum
  differential games,'' \emph{IEEE/CAA J. Autom. Sin.}, vol.~1, no.~3, pp.
  239--247, Jul. 2014, {S}pecial Issue on Extensions of Reinforcement Learning
  and Adaptive Control.  \url{http://ieeexplore.ieee.org/document/7004681}
\BIBentrySTDinterwordspacing

\bibitem{SCC.Luo.Wu.ea2014}
B.~Luo, H.-N. Wu, T.~Huang, and D.~Liu, ``Data-based approximate policy
  iteration for affine nonlinear continuous-time optimal control design,''
  \emph{Automatica}, 2014.

\bibitem{SCC.Kamalapurkar.Walters.ea2016}
\BIBentryALTinterwordspacing
R.~Kamalapurkar, P.~Walters, and W.~E. Dixon, ``Model-based reinforcement
  learning for approximate optimal regulation,'' \emph{Automatica}, vol.~64,
  pp. 94--104, Feb. 2016.
  \url{http://www.sciencedirect.com/science/article/pii/S0005109815004392}
\BIBentrySTDinterwordspacing

\bibitem{SCC.Bian.Jiang2016}
T.~Bian and Z.-P. Jiang, ``Value iteration and adaptive dynamic programming for
  data-driven adaptive optimal control design,'' \emph{Automatica}, vol.~71,
  pp. 348--360, 2016.

\bibitem{SCC.Kamalapurkar.Rosenfeld.ea2016}
\BIBentryALTinterwordspacing
R.~Kamalapurkar, J.~A. Rosenfeld, and W.~E. Dixon, ``Efficient model-based
  reinforcement learning for approximate online optimal control,''
  \emph{Automatica}, vol.~74, pp. 247--258, Dec. 2016.
  \url{http://www.sciencedirect.com/science/article/pii/S0005109816303272}
\BIBentrySTDinterwordspacing

\bibitem{SCC.Kamalapurkar.Reish.ea2017}
\BIBentryALTinterwordspacing
R.~Kamalapurkar, B.~Reish, G.~Chowdhary, and W.~E. Dixon, ``Concurrent learning
  for parameter estimation using dynamic state-derivative estimators,''
  \emph{IEEE Trans. Autom. Control}, 2017, to appear.
  \url{http://ieeexplore.ieee.org/document/7858671/}
\BIBentrySTDinterwordspacing

\bibitem{SCC.Parikh.Kamalapurkar.easubmitteda}
A.~Parikh, R.~Kamalapurkar, and W.~E. Dixon, ``Integral concurrent learning:
  {A}daptive control with parameter convergence without {PE} or state
  derivatives,'' 2017, submitted, see {arXiv}:1512.03464, automatica.

\bibitem{SCC.Kamalapurkar2017}
R.~Kamalapurkar, ``Online output-feedback parameter and state estimation for
  second order linear systems,'' in \emph{Proc. Am. Control Conf.}, 2017, to
  appear, see also, {arXiv:1609.05879}.

\bibitem{SCC.Dinh.Kamalapurkar.ea2014}
\BIBentryALTinterwordspacing
H.~T. Dinh, R.~Kamalapurkar, S.~Bhasin, and W.~E. Dixon, ``Dynamic neural
  network-based robust observers for uncertain nonlinear systems,''
  \emph{Neural Netw.}, vol.~60, pp. 44--52, Dec. 2014.
  \url{http://www.sciencedirect.com/science/article/pii/S089360801400166X}
\BIBentrySTDinterwordspacing

\bibitem{SCC.Xian.Queiroz.ea2004}
B.~Xian, M.~S. de~Queiroz, D.~M. Dawson, and M.~McIntyre, ``A discontinuous
  output feedback controller and velocity observer for nonlinear mechanical
  systems,'' \emph{Automatica}, vol.~40, no.~4, pp. 695--700, 2004.

\bibitem{SCC.Chowdhary2010}
G.~Chowdhary, ``Concurrent learning for convergence in adaptive control without
  persistency of excitation,'' Ph.D. dissertation, Georgia Institute of
  Technology, Dec. 2010.

\bibitem{SCC.Khalil2002}
H.~K. Khalil, \emph{Nonlinear Systems}, 3rd~ed.\hskip 1em plus 0.5em minus
  0.4em\relax Upper Saddle River, NJ: Prentice Hall, 2002.

\end{thebibliography}

\end{document}